\documentclass{article}
\usepackage{epsfig}
\usepackage{subfigure}
\usepackage{color}
\usepackage{amsmath,amssymb,amscd,amsthm} 
\title{Dynamics of conservative peakons in a system of Popowicz}  
\author{L.E. Barnes and A.N.W. Hone\footnote{Currently on leave at the School of Mathematics and Statistics, University of New South Wales, Kensington NSW 2052, Australia.}\\
School of Mathematics, Statistics and Actuarial Science\\ 
Sibson Building, University of Kent\\
Canterbury CT2 , 
UK
}

\newtheorem{rema}{Remark}
\newcommand{\br}{\begin{rem}}
\newcommand{\er}{\end{rem}}
\newcommand{\bex}{\begin{exa}}
\newcommand{\eex}{\end{exa}}
\newcommand{\bd}{\begin{Def}}
\newcommand{\ed}{\end{Def}}
\newtheorem{theorem}[rema]{Theorem}
\newcommand{\bt}{\begin{theorem}}
\newcommand{\et}{\end{theorem}}
\newtheorem{lemma}[rema]{Lemma}
\newtheorem{corollary}[rema]{Corollary}
\newcommand{\bl}{\begin{lemma}}
\newcommand{\el}{\end{lemma}}
\newcommand{\be}{\begin{equation}}
\newcommand{\ee}{\end{equation}}
\newcommand{\bea}{\begin{eqnarray}}
\newcommand{\eea}{\end{eqnarray}}

\newcommand{\adots}{\mathinner{\mkern2mu\raise1pt\hbox{.}\mkern2mu
\raise4pt\hbox{.}\mkern2mu\raise7pt\hbox{.}\mkern1mu}}

\newcommand{\beq}{\begin{equation}}  
\newcommand{\eeq}{\end{equation}}  
\newcommand{\bear}{\begin{array}}  
\newcommand{\eear}{\end{array}} 
\newcommand\la{{\lambda}}   
   
\newcommand\ka{{\kappa}}

\newcommand\rd{{\mathrm{d}}}  
\newcommand\rb{{\mathrm{b}}}  

\newcommand\sgn{{\mathrm{sgn}}}

\newtheorem{thm}{Theorem}[section]

\newtheorem{rem}[thm]{Remark}

\newtheorem{exa}[thm]{Example}

\theoremstyle{remark}

\newcommand{\R}{{\mathbb R}}

\begin{document} 

\maketitle
 
\begin{abstract}
We consider a two-component Hamiltonian system of 
partial differential equations 
with quadratic nonlinearities introduced by Popowicz, which has the form 
of a coupling between the Camassa-Holm and Degasperis-Procesi equations. Despite having reductions 
to these two integrable partial differential equations, the Popowicz system  itself  is not  integrable.  Nevertheless, 
as  one of the authors showed with Irle, it admits 
distributional solutions of peaked soliton (peakon) type, with the dynamics of $N$ 
peakons being determined by a Hamiltonian system on a phase space of dimension $3N$.      
As well as the trivial case of a single peakon ($N=1$), the case $N=2$ is Liouville 
integrable. We present the explicit solution for the two-peakon dynamics, and describe 
some of the novel features of the interaction of peakons in the Popowicz system. 
\end{abstract} 
 
\section{Introduction} 
For the past 25 years there has been a huge amount of interest 
in partial differential equations (PDEs) which admit peaked soliton solutions, known as 
peakons, with a discontinuous first derivative at the peaks. 
This began with the work of Camassa and Holm \cite{ch}, who found the integrable 
PDE 
\beq\label{cheq} 
u_t + 2\ka u_x-u_{xxt} -uu_{xxx} - 2u_xu_{xx}+3uu_x =0 
\eeq 
in the context of shallow water wave theory. In fact this was a 
rediscovery, since the integrability of the latter equation had already 
been recognized in the work of Fokas and Fuchssteiner on 
hereditary symmetries and recursion operators \cite{ff}. 
However, the pioneering contribution of Camassa and Holm was 
their analysis of the remarkable properties of the solutions 
of  (\ref{cheq}), and in particular the fact that in the absence 
of  linear dispersion ($\ka=0$) it 
has multipeakon solutions of the form 
\beq\label{peaks} 
u(x,t)=\sum_{j=1}^N p_j(t)\, e^{-|x-q_j(t)|}, 
\eeq as well 
as displaying wave breaking, and also (for $\ka>0$) smooth solitons vanishing at spatial infinity. 

The equation (\ref{cheq}) with  $\ka=0$  is the case $\rb=2$ of the 1-parameter family 
\beq\label{bfam} 
m_t + u\, m_x +\rb \,u_x \,m = 0, \qquad m=u-u_{xx}, 
\eeq 
introduced in \cite{dhh} after it was shown that the case $\rb=3$, 
identified by Degasperis and Procesi \cite{dp}, is also integrable 
(linear dispersion can always be removed by a combination of a shift $u\to u\,+\,$const and a Galilean transformation). 
With the inclusion of linear dispersion, the whole b-family 
of equations (\ref{bfam}) was subsequently derived via 
shallow water approximations \cite{dgh, cl}. 
All of the equations in the family %
have at least one  Hamiltonian structure,  given by 
\beq\label{hamstr} 
m_t = {\bf B} \, \frac{\delta H}{\delta m}, 
\eeq 
where (subject to appropriate modifications for $\rb=0,1$)  
\beq\label{bop}  H=\frac{1}{\rb-1}\int m \, \rd x, \qquad\mathrm{with}\quad 
{\bf B} = -{\rb^2} \, m^{1-1/\rb}  \partial_x m^{1/\rb} 
{\cal L}^{-1} m^{1/\rb}\partial_x m^{1-1/\rb}   , \quad {\cal L}=\partial_x-\partial_x^3
\eeq 
and admit 
multipeakon solutions of the form (\ref{peaks}). However, 
 $\rb=2,3$ are the only values for which there is a 
bi-Hamiltonian structure, and these 
correspond to the integrable cases,  
in the sense that the equation (\ref{bfam}) has infinitely many local symmetries 
for these values of $\rb$ alone \cite{miknov}. 
 

Due to the discontinuous derivatives at the peaks, it is necessary to specify 
in what sense (\ref{peaks}) is a solution of (\ref{bfam}). The shape of  the peakons 
corresponds to the fact that  $\frac{1}{2} e^{-|x|}$ is the Green's function of 
the one-dimensional Helmholtz operator $1-\partial_x^2$, so for $N$ peakons the quantity 
$m$ is given by a sum of Dirac delta functions, 
\beq\label{deltas} 
m(x,t)=2\sum_{j=1}^N p_j(t) \, \delta (x-q_j(t)), 
\eeq 
with support at   each of the peak positions $x=q_j(t)$ at time $t$. Thus 
it is necessary to interpret (\ref{bfam}) as an equation for distributions. The problem 
is then how to make sense of the nonlinear terms, which include  products of distributions 
with common support. An ad hoc solution  to this problem is to interpret the product $u_xm$ as 
being  $<u_x>m$, where
\beq\label{angle} 
<f(x)> := \frac{1}{2} \lim_{\epsilon\to 0}\Big(f(x+\epsilon)+f(x-\epsilon)\Big)  
\eeq 
is the average of the left and right limits. However, a more satisfying  
solution, which turns out to yield equivalent results,  is the following weak formulation of (\ref{bfam}), 
presented in \cite{hls}: 
\beq\label{weak} 
E[u(x,t)]:=(1-\partial_x^2)u_t 
+ (\rb +1 -\partial_x^2)\partial_x\Big(\frac{1}{2}u^2\Big) + 
\partial_x\Big(\frac{3-\rb}{2}u_x^2\Big)=0.
\eeq 
With the above, $u(x,t)$ is said to be a weak solution if 
\beq \label{weakeq} 
\int E[u(x,t)] \, \phi(x) \, \rd x = 0 
\eeq for all compactly supported test functions $\phi\in C^\infty (\R)$, 
with $\partial_x$ in (\ref{weak}) being viewed as a distributional derivative, 
where it is further required that, for each fixed $t$,  $u_t$ is a 
continuous linear functional, and also $u\in H^1_{\mathrm{loc}}(\R)$, so 
that $u^2$ and $u_x^2$ define  continuous linear functionals as well.   

The preceding requirements entail that (\ref{peaks}) is a weak solution of (\ref{bfam}) 
if and only if $(q_j,p_j)_{j=1,\ldots,N}$ satisfy the system of ordinary differential 
equations (ODEs) 
\beq\label{bodes} 
\frac{\rd q_j}{\rd t} = \frac{\partial \tilde{H}}{\partial p_j}, \qquad  
\frac{\rd p_j}{\rd t} = -(b-1)\frac{\partial \tilde{H}}{\partial q_j},\qquad j=1,\ldots N,  
\eeq 
with 
\beq\label{ht} 
\tilde{H} = \frac{1}{2} \sum_{j,k=1,\ldots N} p_jp_k e^{-|q_j-q_k|}.
\eeq 
In the case $\rb=2$, the ODEs (\ref{bodes}) form a 
canonical Hamiltonian system, with $\tilde{H}$ being the Hamiltonian. 
For all other values of $\rb$,  $\tilde{H}$ is not a conserved quantity; nevertheless, 
for all $\rb$ the equations  (\ref{bodes}) are Hamiltonian with respect to a 
non-canonical Poisson bracket derived by restriction of the bracket defined 
by the operator ${\bf B}$ in (\ref{bop}) to the finite-dimensional 
submanifold of $N$-peakon solutions \cite{peakpuls}. 

The pioneering work of Camassa and Holm inspired the search for integrable 
analogues of (\ref{cheq}) with two or more components, starting with 
\cite{clz, falqui}. 
The subject of this article is the two-component 
system of PDEs given by 
\beq\label{pop} 
\begin{array}{lcl} 
m_t +(2u+v)\, m_x + 3(2u_x+v_x)\, m& = & 0, \quad m=u-u_{xx}, \\
n_t +(2u+v)\, n_x + 2(2u_x+v_x)\, n& = & 0,   \quad n=v-v_{xx},
\end{array} 
\eeq 
which was derived by Popowicz via Dirac reduction of a 
Hamiltonian operator depending on three fields \cite{pop}. With ${\bf m} =(m,n)^T$, the Hamiltonian structure 
of (\ref{pop}) is given by 
\beq\label{popham} 
 {\bf m}_t = \hat{{\bf B}}\, \frac{\delta H_0}{\delta {\bf m}},
\eeq 
with 
\beq\label{popop}  H_0=\int (m+n) \, \rd x, \qquad 
\hat{{\bf B}} = - \left(\begin{array}{cc} 9  m^{2/3}  \partial_x m^{1/3} 
{\cal L}^{-1} m^{1/3}\partial_x m^{2/3} 
& 6  m^{2/3}  \partial_x m^{1/3} 
{\cal L}^{-1} n^{1/2}\partial_x n^{1/2} 
\\ 
6  n^{1/2}  \partial_x n^{1/2} 
{\cal L}^{-1} m^{1/3}\partial_x m^{2/3} 
& 4  n^{1/2}  \partial_x n^{1/2} 
{\cal L}^{-1} n^{1/2}\partial_x n^{1/2} 
\end{array}\right)   . 
\eeq 
The Hamiltonian  operator  $\hat{{\bf B}}$ admits two 
one-parameter families of Casimir functionals, given by 
$$ 
\begin{array}{rcl} 
H_1 & = & \int (nm^{-2/3})^\la m^{1/3}\, \rd x, \\ 
H_2 & = &  \int (nm^{-2/3})^\la (-9n_x^2n^{-2}m^{1/3} 
+ 12 n_xm_xn^{-1}m^{-4/3} -4 m_x^2 m^{-7/3})\, \rd x , 
\end{array} 
$$ 
where $\la$ is arbitrary. The system (\ref{pop}) is a coupling between the Camassa-Holm 
and Degasperis-Procesi equations, that is the cases $\rb=2,3$ of (\ref{bfam}), to which it 
reduces when $u=0$, $v=0$, respectively, and this led Popowicz to speculate that it 
should be integrable. However, a combination of a reciprocal transformation together with Painlev\'e 
analysis, applied by one of us in work with Irle \cite{hi}, provides strong evidence of the non-integrability 
of the coupled system    (\ref{pop}). 

Despite its apparent non-integrability, it is nevertheless the case that the 
Popowicz system admits multipeakon solutions, given by the ansatz 
\beq\label{popeak}
u(x,t) = \sum_{j=1}^N a_j(t)\, e^{-|x-q_j(t)|}, 
\qquad 
v(x,t) = \sum_{j=1}^N b_j(t)\, e^{-|x-q_j(t)|}, 
\eeq 
whose properties were outlined in 
\cite{hi}. The purpose of this article is to describe 
more precisely in what sense %
these are distributional solutions of (\ref{pop}), and provide some details 
of the dynamics of the peakons, which behave somewhat differently from those that appear in 
the Camassa-Holm 
and Degasperis-Procesi equations. Due to the Hamiltonian properties of the solutions 
(\ref{popeak}), which are inherited from those of the PDE system, we refer to them as conservative 
peakons, following \cite{jacekxiangke}, where peakons with analogous properties were considered 
for a family of peakon equations derived from the bi-Hamiltonian structure of the nonlinear Schr\"odinger hierarchy.

\section{Conservative peakons} 

\setcounter{equation}{0}

The first thing to observe about the system (\ref{pop}) is that it does not admit a weak formulation 
suitable for  multipeakon solutions of the form (\ref{popeak}),   
analogous to the formulation (\ref{weak}) for the b-family of equations. If the two coupled 
equations  
are   denoted by $E_j[u(x,t),v(x,t)]=0$, $j=1,2$, 
then a bona fide weak solution should be one for which 
$$ \int (E_1[u(x,t),v(x,t)] \, \phi_1(x)+ E_2[u(x,t),v(x,t)] \, \phi_2(x)) \, \rd x = 0, $$
for an arbitrary pair of compactly supported test functions $\phi_1,\phi_2\in C^\infty(\R)$, 
with sufficiently many derivatives in $E_1$, $E_2$ being interpreted as distributional derivatives. 
For peakons, products of $u,v$ and $u_x,v_x$ define continuous linear functionals, 
and all higher derivatives should be viewed in the sense of distributions. If 
we start by considering $E_1$, then we can use (\ref{weak}) to write this as 
$$ 
E_1=(1-\partial_x^2)
u_t +
2 (4-\partial_x^2)\partial_x \Big(\frac{1}{2}u^2\Big)+ R, 
$$ 
where 
\beq\label{rdef} R=(u_x-u_{xxx}) v +3(u-u_{xx})v_x.\eeq  
There are four mixed products of $u,v$ and their first derivatives, so we need  
to be able to write the terms in $R$ with a total of three $x$ derivatives  in the form 
\beq \label{rdes} 
A\partial_x^3(uv)+ \partial_x^2(Buv_x+Cu_x v) + D\partial_x(u_x v_x), 
\eeq 
for some constants $A,B,C,D$, where $\partial_x$ above 
should be regarded as a distributional derivative. 
(Strictly speaking, isolated products $uv_x$ and $u_xv$ have  discontinuities in the case of 
peakons, but we write these terms separately for the sake of completeness.)  
Upon comparing the coefficients of $u_{xxx}v$ and $u_{xx}v_x$ (and 
the absent terms $u_xv_{xx}$, $uv_{xxx}$) in 
(\ref{rdef}) with (\ref{rdes}), we find  the linear system 
$$ 
\begin{array}{rcl} 
A+C & = & -1, \\ 
3A+B+2C+D& = & -3, \\ 
3A +2B+C+D & = & 0, \\ 
A+B & = & 0, 
\end{array} 
$$ 
which has no solution, so there can be no weak formulation suitable for peakons. 

Despite the fact that the Popowicz system does not admit a weak formulation for multipeakons, 
one can select a distributional interpretation, 
by using the average (\ref{angle}), 
in such a way that the Hamiltonian properties of the PDE system are inherited by these solutions. By 
taking the common prefactor in $(2u_x+v_x)m$ and $(2u_x+v_x)n$ to mean the average $<2u_x+v_x>$, 
one finds a system of equations for distributions, namely  
\beq\label{popd} 
\begin{array}{lcl} 
m_t +\partial_x \Big((2u+v)\, m\Big) + 2<2u_x+v_x>\, m& = & 0,  \\
n_t +\partial_x \Big((2u+v)\, n\Big) + <2u_x+v_x>\, n& = & 0,   
\end{array} 
\eeq 
with $\partial_x$ being the distributional derivative. For peakons, the main upshot of this averaging procedure is that, 
to include the situation where it appears in front of a delta function with the same support, 
the derivative of $e^{-|x|}$ should interpreted as $-\sgn(x)e^{-|x|}$, where the signum function is defined by 
$$ 
\sgn(x) =\begin{cases} 1, &\text{if  $x>0$}; \\
0, & \text{if $x=0$} ; \\ 
-1, &\text{if $x<0$.} 
\end{cases} 
$$   
 The quantities $m$, $n$ are defined as above, so that for 
multipeakons of the form   (\ref{popeak})  they are given by 
\beq\label{mndelta} 
m(x,t) = 2\sum_{j=1}^N a_j(t) \, \delta(x-q_j(t)), \qquad 
n(x,t) = 2\sum_{j=1}^N b_j(t) \, \delta(x-q_j(t)),
\eeq 
Making use of the nomenclature from  \cite{jacekxiangke}, it is appropriate to refer to the multipeakons  which satisfy (\ref{popd}), in the sense of distributions, as 
conservative peakons, since (in particular) the Hamiltonian functional $H_0$ is conserved by these 
solutions.  Furthermore, by rewriting the first order differential operators appearing in  (\ref{popop}) as
$$
 m^{2/3}\partial_x m^{1/3}=m \partial_x +\frac{m_x}{3}, \quad 
 m^{1/3}\partial_x m^{2/3}=m \partial_x +\frac{2m_x}{3}, \quad
 n^{1/2}\partial_x n^{1/2}=n \partial_x +\frac{n_x}{2},
$$ 
to remove the fractional powers (which do not make sense 
for distributions),  the Poisson structure defined by
$\hat{{\bf B}}$ can be reduced to these multipeakon solutions. The following result was stated without proof in \cite{hi}. 

\begin{theorem}\label{maint} With the formulation (\ref{popd}), the Popowicz system admits $N$-peakon 
solutions of the form (\ref{popeak}), where 
the amplitudes $a_j$, $b_j$ and positions $q_j$ satisfy the dynamical system 
\beq\label{peakoneq} 
\begin{array}{rcl} 
\dot{a}_j & = & 2a_j\sum_{k=1}^N (2a_k +b_k)\sgn(q_j-q_k)e^{-|q_j-q_k|}, \\  
\dot{b}_j & = & b_j\sum_{k=1}^N (2a_k +b_k)\sgn(q_j-q_k)e^{-|q_j-q_k|}, \\  
\dot{q}_j & = & \sum_{k=1}^N (2a_k +b_k)e^{-|q_j-q_k|}, \\  
\end{array}  
\eeq
for $j = 1, \ldots ,N$. These equations are in Hamiltonian form, that is 
$$ \dot{a}_j = \{ \, a_j, h \, \}, \qquad \dot{b}_j = \{ \, b_j, h \, \}, \qquad\dot{q}_j = \{ \, q_j, h \, \}, $$ 
with the Hamiltonian 
\beq\label{h} 
h=2\sum_{j=1}^N (a_j + b_j)  
\eeq 
and the Poisson bracket
\beq\label{pbr} 
\begin{array}{rcl} 
\{ \, a_j, a_k \, \} & = & 2a_ja_k \,\sgn (q_j-q_k)  e^{-|q_j-q_k|}, \\ 
\{ \, b_j, b_k \, \} & = & \frac{1}{2} b_jb_k \,\sgn (q_j-q_k)  e^{-|q_j-q_k|}, \\ 
\{ \, q_j, q_k \, \} & = & \frac{1}{2} \sgn (q_j-q_k) (1-  e^{-|q_j-q_k|}), \\ 
\{ \, q_j, a_k \, \} & = & a_k  \, e^{-|q_j-q_k|}, \\ 
\{ \, q_j, b_k \, \} & = & \frac{1}{2}b_k\, e^{-|q_j-q_k|}, \\ 
\{ \, a_j, b_k \, \} & = & a_jb_k \,\sgn (q_j-q_k)  e^{-|q_j-q_k|},   
\end{array}  
\eeq 
where the latter has $N$ Casimirs given by 
\beq\label{cas} 
C_j = \frac{a_j}{b_j^2},  \qquad \text{for $b_j\neq 0$}, \qquad j=1,\ldots, N. 
\eeq 
\end{theorem} 
\begin{proof} By substituting (\ref{popeak}) and (\ref{mndelta}) into  (\ref{popd}) and 
integrating against test functions $\phi_1$, $\phi_2$ with support in a small neighbourhood of 
$x=q_j$, such that $\phi_1(q_j)=1$,  $\phi_1'(q_j)=0$ and $\phi_2(q_j)=0$,  $\phi_2'(q_j)=1$, 
one obtains the equations 
$$ 
\dot{a}_j +2a_j\Big(<2u_x(q_j)>+<v_x(q_j)>\Big)=0, \qquad 
\dot{b}_j +b_j\Big(<2u_x(q_j)>+<v_x(q_j)>\Big)=0, 
$$
$$
\dot{q}_j -\Big(2u(q_j)+v(q_j)\Big)=0,    
$$ 
which yield (\ref{peakoneq}). The Hamiltonian (\ref{h}) is obtained 
by inserting (\ref{mndelta}) into the functional $H_0$ and integrating. 
The derivation of the Poisson brackets (\ref{pbr}) follows the same steps 
as applied to the case of the b-family peakons in \cite{peakpuls}: one 
starts from the expressions for local brackets between fields,  
defined by the Hamiltonian operator $\hat{{\bf B}}$ for the PDE,  
as in (\ref{popop}).   
For instance, with $G(x)=\frac{1}{2}\sgn(x)(1-e^{-|x|})$ one has 
$$ 
\{ \, m(x), m(y)\, \} = m_x(x)m_x(y)G(x-y)
+3(m(x)m_x(y) - m_x(x)m(y))G′(x-y) - 9m(x)m(y)G′′(x-y),
$$  
and then substituting in (\ref{mndelta}) on both sides and  integrating against pairs of  test functions 
of $x$ and $y$,
of the same of form as  
$\phi_1$, $\phi_2$ above, with support at $x=q_j$ and $y=q_k$ for each pair $j,k$, produces 
the brackets $\{ \, a_j, a_k \, \} $, $\{ \, q_j, a_k \, \}$ and $\{ \, q_j, q_k \, \}$, 
while the other brackets are derived from the 
 expressions for $\{ \, m(x), n(y)\, \} $ and $\{ \, n(x), n(y)\, \} $ given in \cite{hi}; 
further details of this calculation can be found in \cite{irle}. 
It is straightforward to check that, away from where the $b_j$ vanish, each $C_j$ is a
Casimir for the bracket specified by (\ref{pbr}), and this is a complete set 
of Casimirs since the bracket has rank $2N$. 
\end{proof} 

In \cite{clp} it was shown that, in addition to the quantity $\tilde{h}=\sum_{j=1}^N p_j$ (which, 
up to rescaling, corresponds to the restriction  of the functional $H$ in (\ref{bop}) to the multipeakon solutions), 
the equations (\ref{bodes}) for the peakons in the b-family, ordered so that $q_j<q_{j+1}$ for 
all $j$,  admit the first integral 
$$ 
P = \left(\prod_{j=1}^Np_j\right) \prod_{k=1}^{N-1}(1-e^{-|q_k-q_{k+1}|})^{\rb-1}.
$$ 
An analogous result holds for the peakons in the Popowicz system, as was noted in \cite{hi} in the case $N=2$. 

\bl\label{j1st} 
In addition to the Hamiltonian $h$ and the Casimirs $C_j$, $j=1,\ldots,N$, the ODEs 
(\ref{peakoneq}) for the 
peakons in the Popowicz system admit the first integral 
\beq\label{jint} 
J =  \left(\prod_{j=1}^Nb_j\right) \prod_{k=1}^{N-1}\left(1-e^{-|q_k-q_{k+1}|}\right), 
\eeq 
so that $\{\,h,J\,\}=0$, 
where the peaks are ordered as follows: 
\beq\label{ordering} 
q_1<q_2<\cdots < q_N.
\eeq
\el
\begin{proof}
Taking the logarithm of (\ref{jint}) and differentiating gives 
$$ 
\frac{\rd}{\rd t} \log J = \sum_{j=1}^N \frac{\rd}{\rd t} \log b_j +\sum_{k=1}^{N-1} \frac{(\dot{q}_k - \dot{q}_{k+1})\sgn(q_k-q_{k+1}) E_{k,k+1}}{1-E_{k,k+1}}, 
$$  
where we have introduced the convenient notation 
$$ 
E_{j,k}=E_{k,j}=e^{-|q_j-q_k|}. 
$$ 
Substituting for the time derivatives from (\ref{peakoneq}) yields 
$$ 
\begin{array}{rcl} 
\frac{\rd}{\rd t} \log J & = & \sum_{j,k=1}^N (2a_k + b_k)\sgn(q_j-q_k) E_{jk} +\sum_{k=1}^{N-1}\sum_{\ell=1}^N (2a_\ell + b_\ell)\sgn(q_k-q_{k+1}) \frac{ (E_{k,\ell}-E_{k+1,\ell}) E_{k,k+1}}{1-E_{k,k+1}} \\
 & = & \sum_{k=1}^N (2a_k + b_k)\, S_k, 
\end{array}  
$$
where,  with the ordering (\ref{ordering}), 
$$ 
S_k =-\sum_{j=1}^{k-1} E_{jk} + \sum_{j=k+1}^{N} E_{jk} - \sum_{\ell=1}^{N-1} \frac{(E_{\ell,k}-E_{\ell+1,k})E_{\ell,\ell+1}}{1-E_{\ell,\ell+1}}. 
$$  
Then the properties of the exponential, together with the assumed ordering of the peakons, produce the identity  
$$ 
\frac{(E_{\ell,k}-E_{\ell+1,k})E_{\ell,\ell+1}}{1-E_{\ell,\ell+1}} =\begin{cases} 
-E_{\ell,k}, &\text{for  $1\leq\ell\leq k$}; \\
E_{\ell+1,k}, &\text{for $k\leq \ell\leq N-1$.} 
\end{cases} 
$$ 
Thus $S_k=0$ for all $k$, and the result follows. 
Note that $\rd f/\rd t = \{\, f,h\,\}$ for any function on phase space, 
so $J$ Poisson commutes with $h$.
\end{proof} 

In the case of a single peakon ($N=1$), the ODE system  (\ref{peakoneq}) is trivially integrable, 
and the fields $u$, $v$ take the form 
\beq\label{N1}
u(x,t) = a e^{-|x-ct-x_0|}, \qquad  v(x,t) = b e^{-|x-ct-x_0|}, \qquad \mathrm{with} \quad c=2a+b, 
\eeq 
where $a,b,x_0$ are arbitrary constants. 
For the case $N=2$, upon  restricting to four-dimensional symplectic leaves $C_1=$ const, $C_2=$ const, 
by the preceding result   there remain 
the two independent first integrals $h$, $J$ with $\{\,h,J\,\}=0$, hence we have the following 

\begin{corollary}\label{N2} 
The Hamiltonian  system  (\ref{peakoneq})  is Liouville  integrable for $N=1$ and $N=2$. 
\end{corollary}    

For the b-family, in each of the special cases $\rb=2,3$ there is a linear system (Lax pair) which can be used to 
construct $N$ independent first integrals for the corresponding ODE system (\ref{bodes}) \cite{ch, dhh}, and 
can be further employed to develop a spectral theory for the peakons, leading to an 
explicit solution for all $N$ \cite{bss, ls}. However, for $N>2$ there is no reason to expect that 
the system (\ref{peakoneq}) has any first integrals other than $h$, $J$ and the $N$ Casimirs. 
Thus, while Liouville's theorem guarantees that the solution for $N=2$ can be found by quadratures, 
which we explicitly derive in the next section, for larger $N$ this does not seem possible. 


\section{Explicit dynamics of two peakons} 
\setcounter{equation}{0}

For arbitrary $N$, the $3N$-dimensional system (\ref{peakoneq}) can always be reduced to $2N$-dimensional 
symplectic leaves by fixing the values of the $N$ Casimirs  $C_j$ (away from $b_j=0$); in particular, one 
can eliminate the variables $a_j$ to leave $2N$ equations for $(q_i,b_i)_{i=1,\ldots,N}$. 
Here we consider the Liouville integrable case $N=2$, for which 
there are the two first integrals 
\beq\label{1stint} 
h=2(a_1+a_2+b_1+b_2), \qquad J=b_1b_2(1-e^{-|q_1-q_2|}), 
\eeq 
in addition to the two Casimirs $C_1$, $C_2$,  and show how to explicitly integrate the equations of 
motion. For the sake of concreteness, we restrict to the situation where $a_1,a_2>0$, so 
that the values of the Casimirs  are positive, and fix these to be constant values: 
\beq\label{casik} 
C_i=k_i^2, \qquad k_i>0, \qquad i=1,2. 
\eeq 
We further assume that $b_1,b_2>0$ (at least at time $t=0$),  meaning that 
in this case both fields $u$ and $v$ initially consist of 
peakons, with positive amplitudes (rather than anti-peakons, with negative amplitudes); as we 
shall see, this implies that the amplitudes $b_i$ remain positive for all time. 
Thus we can reduce the solution of (\ref{peakoneq}) for $N=2$ to solving the system 
\beq\label{bq} 
\begin{array}{rcl} 
\dot{b}_1 &=& b_1b_2(2k_2^2b_2+1)\sgn(q_1-q_2)e^{-|q_1-q_2|}, \\ 
\dot{b}_2&=& b_1b_2(2k_1^2b_1+1)\sgn(q_2-q_1)e^{-|q_1-q_2|}, \\ 
\dot{q}_1 &=& b_1(2k_1^2b_1+1) + b_2(2k_2^2b_2+1)e^{-|q_1-q_2|}, \\ 
\dot{q}_2 &=& b_2(2k_2^2b_2+1) + b_1(2k_1^2b_1+1)e^{-|q_1-q_2|}, 
\end{array}   
\eeq 
and then the amplitudes $a_i$ are determined by $a_i=k_i^2b_i^2$, $i=1,2$. 

Following  \cite{ch}, it is convenient to set 
$$
q=q_1-q_2, \qquad Q=q_1+q_2, 
$$ 
and we will assume that the peaks are ordered so that 
\beq\label{order} 
q_1<q_2\implies q<0; 
\eeq 
if the latter condition holds initially, then it  will continue to hold as long as the peakons do not 
overlap (this  possibility will be considered in due course). 
In that case, the system (\ref{bq}) is equivalent to 
\beq\label{Qq} 
\begin{array}{rcl} 
\dot{b}_1 &=& -b_1b_2(2k_2^2b_2+1)e^{q}, \\ 
\dot{b}_2&=& b_1b_2(2k_1^2b_1+1)e^{q}, \\ 
\dot{q} &=&\Big( b_1(2k_1^2b_1+1) - b_2(2k_2^2b_2+1)\Big) (1-e^{q}), \\ 
\dot{Q} &=& \Big( b_1(2k_1^2b_1+1) + b_2(2k_2^2b_2+1)\Big) (1+e^{q}). 
\end{array}   
\eeq 
In order to integrate the above equations explicitly, it is useful to note that, substituting for 
$a_1,a_2$ in terms of $b_1,b_2$ and 
fixing the Hamiltonian $h$ to a constant value defines an ellipse in the $(b_1,b_2)$ 
plane, i.e.\
\beq\label{ellipse} 
h = 2b_1(k_1^2b_1+1) + 2b_2(k_2^2b_2+1)=\mathrm{const}, 
\eeq 
which can be specified parametrically in terms of an angle 
$\theta \in (-\pi,\pi]$, so that $b_1$ and $b_2$ are given by 
\beq\label{b12theta}
b_1 = \frac{\la}{k_1} \sin \theta - \frac{1}{2k_1^2}, \qquad 
 b_2 = \frac{\la}{k_2} \cos \theta - \frac{1}{2k_2^2}, \qquad \mathrm{where} \quad
\la^2 = \frac{h}{2}+\frac{1}{4k_1^2} + \frac{1}{4k_2^2}, \eeq 
with $\theta$  measured clockwise from the vertical, that is, the positive $b_2$ axis. 

\begin{lemma} \label{pos} 
If the initial amplitudes $b_1(0)$, $b_2(0)$ are positive, then they remain positive for all time $t$, 
and the peakons do not overlap. 
\end{lemma}
\begin{proof} Since $(b_1(t),b_2(t))$ lies on the ellipse $h=$ const defined by (\ref{ellipse}), 
the amplitudes $b_j$ are bounded  for all $t$. From the formula for $J$ in  (\ref{1stint}),  
the assumption $q(0)=q_1(0)-q_2(0)<0$ implies $J>0$. However, $q(t)=0$ for some $t$ would imply 
$J=0$, contradicting the fact that $J$ is a first integral, so the two peaks cannot overlap and 
$q(t)<0$ for all $t$. Thus 
\beq\label{jform}  
J=b_1b_2(1-e^q)<b_1b_2, 
\eeq    
so $(b_1(t),b_2(t))\in\R^2$ lies in the positive quadrant above the upper branch of the hyperbola $b_1b_2=J$.  
\end{proof} 

\begin{rema} Adapting an argument used in \cite{clp}, Proposition 2.4,  
in the case where the amplitudes are positive we may write 
$$ 
-hb_1< -(2k_2^2b_2^2+b_2)b_1\leq \dot{b}_1 \leq 
(2k_2^2b_2^2+b_2)b_1<hb_1, 
$$
and similarly for $b_2$, so that  both an upper and a lower bound  is obtained 
 for $b_j(t)$, namely 
$$  
b_j(0)\, e^{-ht} \leq b_j(t) \leq b_j(0)\, e^{ht}, \qquad j=1,2
$$ 
holds for all $t\geq 0$, by 
Gronwall's inequality. 
\end{rema}

The first two equations in (\ref{Qq}) both yield the same equation for the time derivative of 
$\theta$, that is  
\beq\label{thetacalc}
\dot{\theta} = -2k_1k_2 b_1b_2e^q = 2k_1k_2(J-b_1b_2), 
\eeq 
(where the second equality comes from fixing the first integral $J=$ const, as in (\ref{1stint}), to eliminate $q=q_1-q_2$).  
Then replacing $b_1,b_2$ with their parametric forms (\ref{b12theta}) in terms of $\theta$ leads 
to an autonomous equation for $\theta$ alone, namely 
\beq\label{thdot} 
\dot{\theta}= 2\Big( 
Jk_1k_2 - \frac{1}{4k_1k_2} +\frac{\lambda}{2k_1} \cos\theta +\frac{\lambda}{2k_2} \sin\theta- \lambda^2 \sin\theta\cos\theta \Big) \equiv f(\theta) 
\eeq  
At this stage we have already shown that   the $N=2$ peakon equations can be 
completely reduced to quadratures (as 
is guaranteed by Liouville's theorem \cite{arnold}). To see this, observe that by performing the quadrature  
$$ 
\int \frac{d\theta}{f(\theta)} = t +\mathrm{const} 
$$ 
 we obtain 
$\theta=\theta(t)$  from (\ref{thdot}),  and then $b_1,b_2$ are specified as functions of $t$ by (\ref{b12theta}); 
hence $q(t)$ is found from 
(\ref{jform}), 
so that  
\beq\label{qform}
q=\log\left(1-\frac{J}{b_1b_2}\right) <0  
\eeq 
by the initial assumption on $\sgn (q)$. Finally, 
having specified the right-hand side of (109) as functions of $t$, an additional quadrature with respect to $t$ 
yields $Q=Q(t)$. 

In order to carry out the integration explicitly, it is convenient to make use of the standard T-substitution, 
to convert  (\ref{thdot}) into a  rational differential equation for the variable 
$$ 
T=\tan\frac{\theta}{2} . 
$$ 
Thus (\ref{thdot}) is transformed to 
\beq\label{Tdot} 
\dot{T}= F(T), \qquad \mathrm{where} \qquad F(T) = \frac{P(T)}{T^2+1},  
\eeq 
with the quartic polynomial 
\beq\label{Pquart}
P(T) = \left( Jk_1k_2 -\frac{1}{4k_1k_2} \right) (T^2+1)^2 
-\frac{\lambda}{2k_1}(T^4-1) + \frac{\lambda}{k_2}T(T^2+1)+2\lambda^2 T(T^2-1). 
\eeq 
Then applying  a partial fraction decomposition, we find 
$$ 
\frac{1}{F(T)}= \frac{T^2+1}{P(T)}=K^{-1} \sum_{j=1}^4 \frac{(T_j^2+1)e_j}{T-T_j}, 
$$ 
where the quartic $P(T)$ is factorized as 
$$ 
P(T)=K\prod_{j=1}^4(T-T_j), \qquad K=Jk_1k_2 -\frac{1}{4k_1k_2}-\frac{\lambda}{2k_1}, \qquad  
\mathrm{and} \quad  
e_j=\prod_{1\leq k \leq 4, k\neq j}(T_j-T_k)^{-1}.
$$ 
Thus the general solution of (\ref{Tdot}) is given implicitly by 
\beq\label{Tsoln} 
K^{-1} \sum_{j=1}^4 (T_j^2+1)e_j \log(T-T_j) = t+\mathrm{const}.
\eeq
The above form of the solution is valid for complex values of $T$ (and $t$), where the constant 
of integration should also be allowed to be complex; but since we 
are interested in real values of $T=\tan(\theta/2)$ for real $t$, in the case where 
the coefficients of $P(T)$ are all real, the solution 
may need to be specified in different forms according to the combinations of real/complex 
roots of this quartic. For example, if the four roots $T_j$ are all real then for 
real $T$ the solution can be written 
as    
\beq\label{realt} 
K^{-1} \sum_{j=1}^4 (T_j^2+1)e_j \log\vert T-T_j\vert = t-t_0 
\eeq
with a real constant of integration $t_0$. If, on the other hand, $P(T)$ has two real roots and 
a complex conjugate pair, then (for real $T$) two of the logarithms in (\ref{Tsoln}) can be 
combined into an arctangent. Note that 
the roots of $P(T)$ correspond precisely to the points in the $(b_1,b_2)$ plane where 
the ellipse (\ref{ellipse}) intersects with the hyperbola $b_1b_2=J$, and 
the proof of Lemma \ref{pos} guarantees that there are two 
such points in the positive quadrant, hence $P(T)$ has  at least two real roots.  
Having obtained $T(t)$ implicitly, from    (\ref{b12theta}) we then find $b_1,b_2$ 
as 
\beq\label{b12form} 
b_1=\frac{2\lambda T}{k_1 (1+T^2)}-\frac{1}{2k_1^2}, \qquad  
b_2=\frac{\lambda (1-T^2)}{k_2 (1+T^2)}-\frac{1}{2k_2^2},
\eeq 
and hence $q$ is found from (\ref{qform}). 

For the second quadrature, to find $Q(t)$ from the last equation in (\ref{Qq}), it is convenient to 
write 
$$ 
\dot{Q}=\frac{dQ}{dT} \, \dot{T}, 
$$ replace the right-hand side of (\ref{Qq}) by the corresponding expressions in terms of $T$, and 
then obtain $Q=Q(T)$ by integrating with respect to $T$ (instead of $t$). This leads to the 
equation 
\beq\label{QTform}
\frac{dQ}{dT}=2\lambda R(T)\, \left(\frac{1}{P(T)}- \frac{1}{\hat{P}(T)}\right),   
\eeq 
given in terms of two additional  polynomials, one quadratic and the other quartic, namely  
$$ 
R(T)=\lambda (T^2+1) 
-\frac{T}{k_1}  +\frac{T^2-1}{2k_2}, \qquad 
\mathrm{and} \qquad  
\hat{P}(T) = Jk_1k_2(T^2+1)^2 -P(T).
$$
Then  
writing 
$$\hat{P}(T)=\hat{K}\prod_{j=1}^4(T-\hat{T}_j)
,$$ 
the general solution 
of (109) can be written in terms of $T=T(t)$ as 
\beq\label{Qsoln} 
Q=2\lambda \Big( K^{-1} \sum_{j=1}^4 R(T_j) e_j \log (T-T_j)  - 
\hat{K}^{-1} \sum_{j=1}^4 R(\hat{T}_j) \hat{e}_j \log (T-\hat{T}_j)\Big) +\mathrm{const}, 
\eeq 
with $$ \hat{e}_j=\prod_{1\leq k \leq 4, k\neq j}(\hat{T}_j-\hat{T}_k)^{-1}. $$

The preceding formulae can be used to describe the scattering of two peakons, 
in terms of their asymptotic behaviour 
as $t\to\pm \infty$. For certain values of the parameters/initial data, 
the behaviour of two peakons in the Popowicz system appears to be  
qualitatively similar to that of peakons in integrable PDEs such as the Camassa-Holm and 
Degasperis-Procesi equations: for large negative/positive times the two peaks are well separated and asymptotically move 
with constant velocities and constant amplitudes. However, there is one main difference: unlike those integrable 
single component equations, in which the two peakons asymptotically switch their velocities and amplitudes, resulting only 
in a phase shift  in the resulting trajectories before/after interaction, the Popowicz peakons exchange different 
amounts of velocity and amplitude during the interaction (with the amplitudes being different for the two 
components $u,v$), so that generically the pair of peakon velocities is different before and after.

The asymptotic form of the two-peakon solution for  the Popowicz system is controlled by the  
first order ODE (\ref{Tdot}) for $T$. This equation has fixed points at the roots of $F(T)$, i.e.\ at $T=T_k$ where 
the roots of the quartic polynomial $P(T)$ lie. Near to a fixed point, the local behaviour is 
$$ 
T\sim T_k +A_k e^{F'(T_k)t}, 
$$
where $A_k$ is a constant, and we have 
\beq\label{fpform} 
F'(T_k) = \frac{K}{(1+T_k^2)e_k} 
\eeq 
compared with the coefficients in (\ref{Tsoln}).
The initial data for the peakon system at $t=0$ determines an initial point on the ellipse $h=$ const in the 
$(b_1,b_2)$ plane, and hence an initial angle $\theta(0)$ and corresponding value $T(0)=\tan(\theta(0)/2)$. 
Given that $T(0)$ lies between two real roots of $F$ (or equivalently of $P$), the fact that $J<b_1b_2$ and the 
assumption $q<0$ implies from (\ref{thetacalc}) that $\dot{\theta}<0$, hence $\dot{T}<0$. We 
denote the two adjacent roots by $T_\pm$ with 
$$T_+<T(0)<T_-,$$ 
and the asymptotic behaviour is then given by 
\beq\label{Tasy} 
T\sim T_\pm \pm e^{F'(T_\pm)t+\delta_\pm} \qquad\mathrm{as}\qquad t\to\pm\infty, 
\eeq 
where the constant $\delta_\pm$  depends on the terms in  (\ref{Tsoln}) that are regular at $T= T_\pm$, as well 
as the integration constant; so in particular, when there are four real roots, $\delta_\pm$ depends on 
the arbitrary constant $t_0$ in (\ref{realt}). Moreover, the given assumptions imply that $T_+$ is a 
stable fixed point of (\ref{Tdot}), and $T_-$ is  
unstable, so 
$$ 
F'(T_+)<0< F'(T_-).
$$ 

From (\ref{b12form}) we can immediately read off the asymptotic amplitudes of the two peakons in the field $v(x,t)$, 
that is  
\beq\label{b12asy} 
b_1\to b_1^{\pm} = \frac{2\lambda T_\pm}{k_1 (1+T_\pm^2)}-\frac{1}{2k_1^2}, \qquad  
b_2\to b_2^{\pm} = \frac{\lambda (1-T_\pm^2)}{k_2 (1+T_\pm^2)}-\frac{1}{2k_2^2},  \qquad\mathrm{as}\quad t\to\pm\infty. 
\eeq 
The corresponding amplitudes for the field $u(x,t)$  are then obtained from the formula $a_j=C_jb_j^2=k_j^2b_j^2$. 
Note that the points $(b_1^{\pm}, b_2^\pm)\in\R^2$ are precisely the two intersections 
of the ellipse (\ref{ellipse}) with the upper branch of the hyperbola $b_1b_2=J$, which must always 
exist if the initial amplitudes are positive, by 
the proof of Lemma \ref{pos}. 
The asymptotic behaviour of the positions is more complicated. For the difference $q=q_1-q_2$ we have from 
(\ref{thetacalc}) and  (\ref{qform})  that 
\beq\label{qasy}
\begin{array}{rcl}
q & = & \log (-\dot{T}) - \log(k_1k_2b_1b_2(1+T^2)) \\
& \sim & F'(T_\pm)t+\delta_\pm+\log(\mp F'(T_\pm)) -\log(k_1k_2J(1+T_\pm^2))  
\quad\mathrm{as}\quad t\to\pm\infty, 
\end{array}
\eeq 
where we used (\ref{Tasy}) and the fact that $b_1b_2\to J$ as $|t|\to\infty$.
The sum $Q=q_1+q_2$ is determined  from (\ref{Qsoln}), which near $T=T_\pm$ 
gives 
$$ 
Q\sim \frac{2\lambda e_\pm R(T_\pm) F'(T_\pm)}{K} \, t + \mathrm{const}, 
$$ 
hence from (\ref{fpform}) we have 
\beq\label{Qasy} 
Q\sim \frac{2\lambda R(T_\pm)}{1+T_\pm^2} \, t +  \mathrm{const} \qquad\mathrm{as}\qquad t\to\pm\infty, 
\eeq 
where the constant depends on $\delta_\pm$, as well as the arbitrary constant of integration 
and the other terms 
in (\ref{Qsoln}). 

It is worth comparing these results with the corresponding asymptotic formulae for Camassa-Holm peakons \cite{ch, chh}:  
in that case, if the leftmost peak at position $q_1$ has asymptotic velocity $c_1$  (which is 
the same as its amplitude) for large negative times, and the peak at $q_2>q_1$ has  asymptotic velocity $c_2$, with $c_1>c_2$ so that they 
collide, then for large positive times these asymptotic velocities (and amplitudes) are switched. In 
terms of the difference $q$  this corresponds to 
having 
$$ 
q\sim \mp (c_1-c_2)t + \mathrm{const} \qquad
\mathrm{as} \quad t\to \pm\infty, 
$$ 
while for the sum $Q$ the leading order behaviour is the same in both asymptotic regimes, that is 
$$ 
Q\sim  (c_1+c_2)t + \mathrm{const} \qquad
\mathrm{as} \quad t\to \pm\infty;  
$$ 
the next to leading order (constant) terms determine the phase shifts, i.e.\ the changes in the relative position 
of each soliton that result from their interaction.  For the Popowicz peakons, in contrast, such 
switching of asymptotic velocities is far from being  generic behaviour: 
from (\ref{qasy}) and (\ref{Qasy}), 
asymptotic switching would require that both 
$$
P'(T_+)/(T_+^2+1) = - P'(T_-)/(T_-^2+1)
$$ 
(the asymptotic velocity of $q$ changes sign between $t\to\pm\infty$)
and 
$$ 
\frac{2\lambda R(T_+)}{1+T_+^2}=\frac{2\lambda R(T_-)}{1+T_-^2}
$$ 
($Q$ has the same asymptotic velocity for $t\to\pm\infty$)  
should hold, which puts  constraints on the parameters/initial values. 

\begin{figure} \centering 
\includegraphics[width=8cm,height=8cm,keepaspectratio]{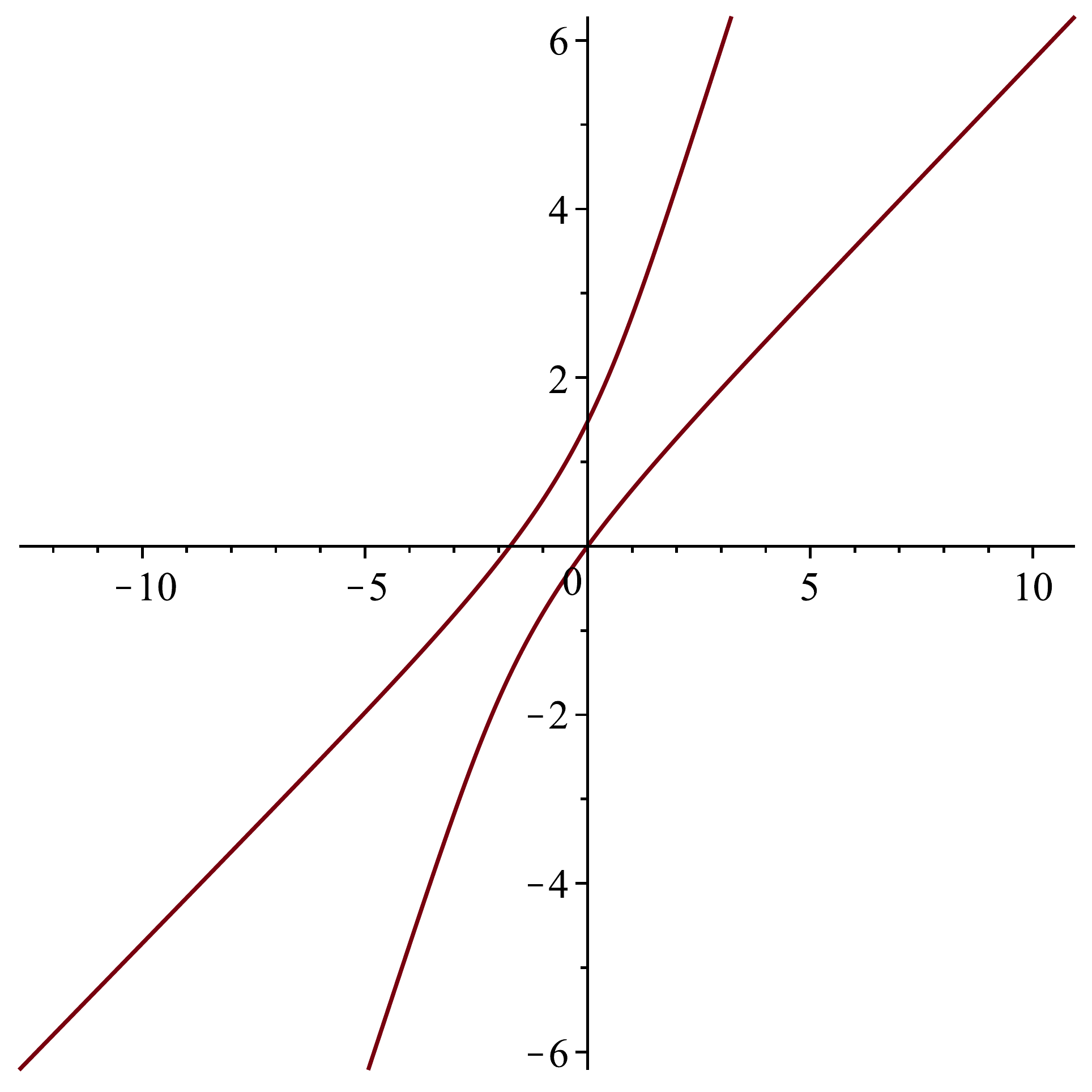}
\caption{\small{Scattering of two peakons with initial data (\ref{inits}) as a spacetime plot (time is the vertical axis).}}
\label{scat}
\end{figure}

To illustrate these results, it is instructive to consider a particular numerical example. 
Upon choosing the initial data and parameters for (\ref{bq}) as 
\beq\label{inits} 
b_1(0) = \frac{94417}{165416}, \, 
 b_2(0) = \frac{103298821}{164092672}, \, 
q_1(0) = \log \left(\frac{2890857125}{16447158149}\right), \, 
q_2(0) = 0, \, k_1=1, \, k_2= \frac{496}{593},  
\eeq
the values of the first integrals are 
found to be 
$$ 
h=\frac{58672865}{16267808}, \qquad J = \frac{298710111}{1008604096}. 
$$  
Then  $\la = 35425/22816$, and 
$$ P(T) =-\frac{3365375}{4066952}\left(T+\frac{111}{152}\right) 
\left(T-\frac{3}{10}\right) 
\left(T-\frac{1}{2}\right) (T-8),  
$$ 
so that $T(0)=2/5$ lies between $T_+=3/10$ and $T_-=1/2$, and also 
$$ 
\hat{P}(T) =\frac{6133}{5704}\left(T^2 - \frac{35425}{5704}T+1 \right)\left(T^2-\frac{10893}{24532}\right), 
$$ 
where the latter quartic  also has four real roots, namely 
$T=\frac{35425}{5704}\pm \sqrt{\frac{1124788161}{32535616}},$ $\pm \sqrt{\frac{10893}{24532}}$. 
Substituting the roots of $P$ 
into (\ref{realt}) and setting $T=2/5$ at $t=0$ yields   
$
t_0\approx 0.2686597887$, 
and similarly the integration constant in (\ref{Qsoln}) can be fixed, 
so that $q_1=(Q+q)/2$ and $q_2=(Q-q)/2$ are completely determined 
parametrically in terms of $t(T)$. Figure \ref{scat} is a plot of their trajectories, with $q_1$ being the 
topmost curve.  Figure \ref{pscat} shows the same scattering process in the form 
of a contour plot of $v=v(x,t)$ viewed from above; the figure was made by joining 
two different parametric plots, which were 
required to deal with larger positive/negative 
$t$ values separately. 
The asymptotic amplitudes are 
$$ 
(b_1^-,b_2^-)=\left(\frac{4233}{5704}, \frac{70567}{176824}\right)\approx 
(0.7421107994, 
0.3990804416) 
\qquad 
\mathrm{for} \quad t\to - \infty, 
$$ 
and 
$$ 
(b_1^+,b_2^+)=\left(\frac{2023}{5704}, \frac{147657}{176824}\right)\approx 
(0.3546633941,  
 0.8350506719) 
\qquad 
\mathrm{for}\quad  t\to \infty.
$$ 
Then combining (\ref{qasy}) and (\ref{Qasy}), we 
find that the asymptotic positions of the two peakons are given by 
$$ 
q_j\sim c_j^\pm t+\mathrm{const}, \qquad \mathrm{as} \quad t\to\pm\infty, \qquad j=1,2,  
$$ 
where for $t\to -\infty$ the asymptotic velocities are 
$$ 
c_1^- = \frac{29990805}{16267808}\approx 1.8435676767 
, \qquad 
c_2^- = \frac{2529345}{4066952}\approx 0.6219264452 
, 
$$  
and for $t\to \infty$ they are 
$$ 
c_1^+ = \frac{9862125}{16267808}\approx 0.6062356404 
, \qquad 
c_2^+ = \frac{7364175}{4066952}\approx 1.8107356566 
. 
$$
Hence in this case the two peakons do not exactly exchange their asymptotic velocities and amplitudes, but only 
approximately so. 

\begin{figure} \centering 
\includegraphics[width=9cm,height=9cm,keepaspectratio]{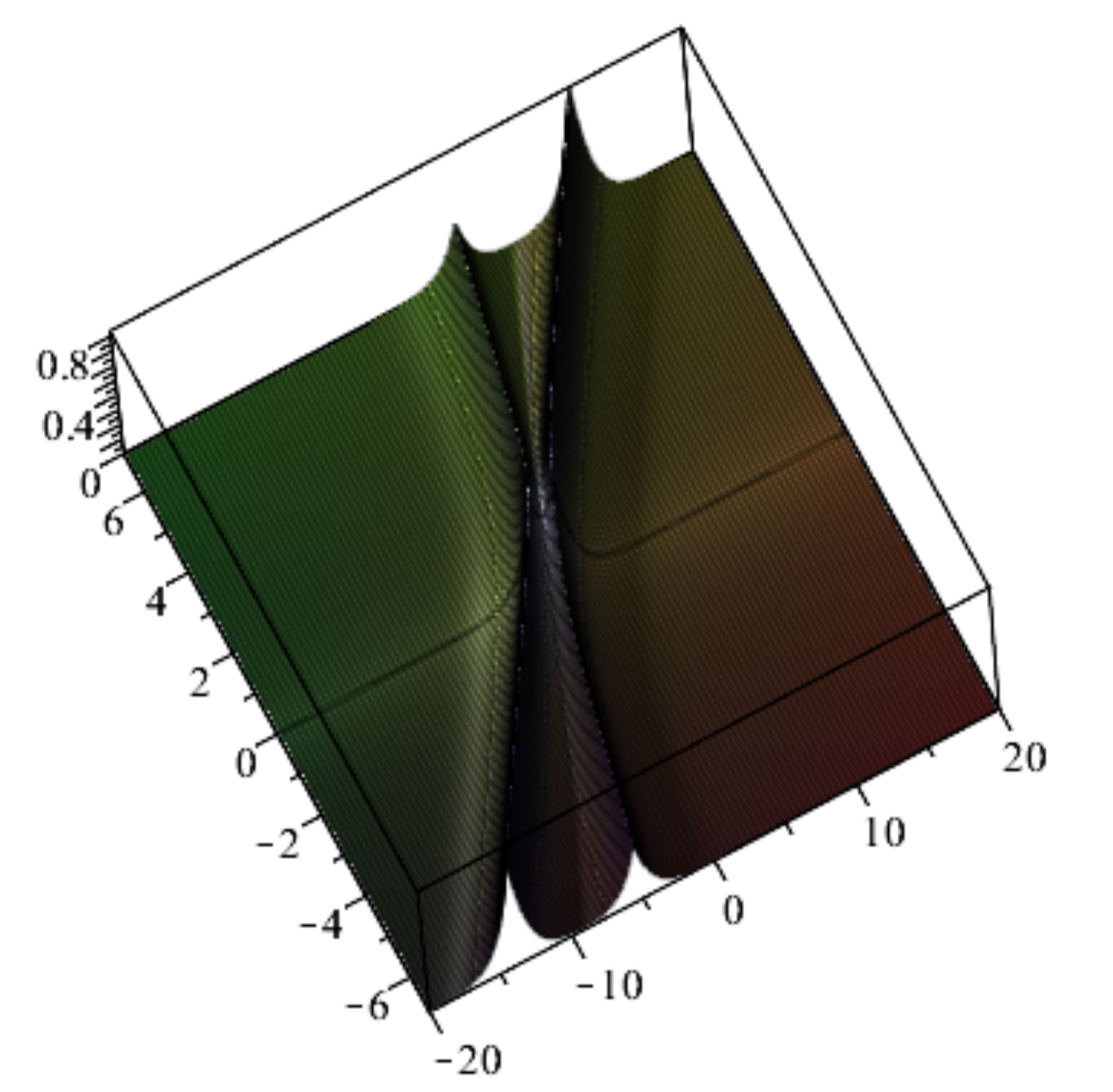}
\caption{\small{
Contour plot of $v(x,t)$ in the box 
$-20\leq x\leq 20$, $-6\leq t\leq 6$, $0\leq v\leq 0.8$.}}
\label{pscat}
\end{figure}

\section{Conclusions}
\setcounter{equation}{0} 

We have considered the dynamics of peakon solutions in the non-integrable coupled system (\ref{pop}). 
In the absence of a weak formulation appropriate for these solutions, they are interpreted as 
distributional solutions in such a way that the peakons inherit the Hamiltonian properties of the PDE 
system, so their dynamics is conservative. The two-peakon dynamics is Liouville integrable, 
and we have explicitly integrated the equations of motion and described the interaction of the 
peakons in the case when all the amplitudes are positive. The case where the amplitudes 
have mixed sign (peakon-antipeakon interaction) is more subtle: in the Camassa-Holm case, 
it involves a head-on collision, with overlapping peaks \cite{ch, chh}; while in this case, 
if the peakons overlap ($q=0$) then the form of first integral $J$, as in (\ref{jform}), implies that 
at least 
one of the amplitudes must diverge to infinity.

For three or more peakons, we do not expect that the dynamics of peakons is integrable. 
Nevertheless, in the case where all the peakons have  positive amplitudes, the qualitative 
features of their interaction 
should be similar to the two-peakon case. In particular, it is not hard to see that 
the analogue of Lemma \ref{pos} holds for all $N$: upon fixing $a_j=k_j^2 b_j^2>0$, 
the fixed  energy hypersurface 
$$ 
h=2\sum_{j=1}^N  (k_j^2b_j^2+b_j) =\mathrm{const} 
$$ 
is a compact quadric  in $\R^N$, so from the form of (\ref{jint}) there can 
be no overlap  $q_k=q_{k+1}$ between initially  adjacent peaks; but then, from the 
ordering (\ref{ordering}), peaks that are initially non-adjacent cannot overlap without first passing through 
their nearest neighbours, which cannot happen, $0<J<\prod_{j=1}^N b_j$ holds,  and the result follows. 
Furthermore, it seems reasonable that the overall dynamics of three or more peakons should be 
determined approximately by the local  interaction between each pair of peaks, at least 
when they are well separated from the rest. 

It would be interesting to carry out numerical studies of the peakon ODEs (\ref{peakoneq})  for $N>2$, and 
to perform a numerical integration of the full PDE system (\ref{pop}) 
to see whether peakons emerge  naturally from generic initial data, as is the case for the b-family 
in the parameter range $\rb>1$ \cite{hs}. However, the PDE integration is likely to be at least as challenging as 
for scalar peakon equations, which are already known to be difficult (for instance, see \cite{clp} and references).   
Before embarking on such a study, it is worth noting that all 
of the considerations in this paper admit a natural generalization to a
vector ${\bf b}$-family of PDEs, given by 
$$ 
{\bf m}_t = \hat{{\bf B}}\, \frac{\delta H_0}{\delta {\bf m}}, 
$$ 
with  
$$ 
{\bf m}=(m_1,m_2,\ldots, m_d)^T, \qquad 
{\bf b}=(\rb_1,\rb_2,\ldots, \rb_d)^T 
$$ 
being a $d$-component vector of fields and a corresponding vector of parameters, and 
$$ 
H_0=\int\sum_{j=1}^d m_j \, \rd x, 
\qquad \hat{{\bf B}} = {\bf w} \, {\cal L}^{-1}\,  {\bf w}^\dagger , 
$$ 
where ${\bf w}=(w_1,w_2,\ldots, w_d)^T$ is a vector operator with components 
$$ 
w_j = \rb_j \, m_j^{1-1/\rb_j}\, \partial_x \, m_j^{1/\rb_j}, \qquad j=1,\ldots,d. 
$$ 
So the original b-family (\ref{bfam}) is just the case $d=1$, while the 
Popowicz system corresponds to $d=2$ with fields $(m_1,m_2)=(m,n)$ 
and parameters $(\rb_1,\rb_2)=(3,2)$.  

\noindent{\bf Acknowledgements:} LEB was supported by a studentship from SMSAS, University of Kent. 
ANWH is supported by EPSRC fellowship EP/M004333/1, and is grateful to the 
School of Mathematics \& Statitsics, UNSW for hosting him as a Visiting 
Professorial Fellow with additional funding from the Distinguished 
Researcher Visitor Scheme. We also thank the reviewers for their comments.


\begin{thebibliography}{99}



\bibitem{jacekxiangke} S.C. Anco, X. Chang and J. Szmigielski, 
{\tt arXiv:1711.01429}
 
\bibitem{arnold} V.I. Arnold, Mathematical Methods of Classical Mechanics, 2nd edition, Springer, 1997.


\bibitem{bss} R. Beals, D.H. Sattinger and J. Szmigielski, 
Adv. 
Math. {\bf 154} (2000) 229--257. 

\bibitem{ch} 
R. Camassa and D.D. Holm, 
Phys. Rev. Lett. {\bf 71} (1993) 1661--1664.

\bibitem{chh} R. Camassa,  D.D. Holm and J.M. Hyman,
Adv. Appl. Mech. {\bf 31} (1994) 1--33. 

\bibitem{clz} M. Chen, S.-Q. Liu and Y. Zhang, 
Lett. Math. Phys. {\bf 75} (2006) 115.

\bibitem{clp} A. Chertock, J.-G. Liu and T. Pendleton, 
SIAM J. Numer. Anal. {\bf 50} (2012) 1--21. 

\bibitem{cl} A. Constantin and D. Lannes, 
Arch. Ration. Mech. Anal. {\bf 192} (2009) 165--186.  

\bibitem{dp} A. Degasperis and M. Procesi, 
Asymptotic integrability, in {\it Symmetry and Perturbation Theory}, 
eds. Degasperis and G.  Gaeta,   River Edge, NJ: World Scientific  (1999) 23--37. 

\bibitem{dhh} A. Degasperis, D.D. Holm and A.N.W. Hone, 
Theoret. Math. Phys. {\bf 133} (2002) 1461--1472. 

\bibitem{dgh} H.R. Dullin, G.A. Gottwald and D.D. Holm, Physica D {\bf 190} (2004) 1--14. 

\bibitem{falqui} 
G. Falqui, 
J. Phys. A: Math. Gen. {\bf 39} (2006) 327--342.

\bibitem{ff}  B. Fuchssteiner  and A.S. Fokas, 
Physica D {\bf 4} 
(1981) 47--66.



\bibitem{peakpuls} D.D. Holm and A.N.W. Hone, 
J. Nonlin. Math. Phys. {\bf 12} (2005) 380--394.

\bibitem{hs} D.D. Holm and M. Staley, 
Phys. Lett. A {\bf 308} 
(2003) 437--444.

\bibitem{hi} A.N.W. Hone and M.V. Irle, 
On the non-integrability of the Popowicz peakon system, in 
{\it Dynamical Systems and Differential Equations, 
Proc. 7th AIMS International Conference},   
Disc. Cont. Dyn. Sys. Supplement (2009) 359--366. 

\bibitem{hls} A.N.W. Hone, H. Lundmark and J. Szmigielski, 
Dyn. Partial Differ. Equ. {\bf 6} (2009) 253--289. 

\bibitem{irle} M.V. Irle, 
{\it Solitons, peakons and Hamiltonian structure of particular partial differential equations}, 
MSc thesis, University of Kent, 2010. 


\bibitem{ls}  H. Lundmark and J. Szmigielski, 
Int. Math. Res. Pap. IMRP {\bf 2005}  (2005) 53--116.


\bibitem{miknov} A.V. Mikhailov and V.S. Novikov, 
J. Phys. A: Math. Gen. {\bf 35} (2002) 4775--4790.

\bibitem{pop} Z. Popowicz, 
J. Phys. A: Math. Gen. {\bf 39} (2006) 13717--13726.



\end{thebibliography}
\end{document}